\documentclass[sigconf]{acmart}

\AtBeginDocument{%
  }

\usepackage{adjustbox}
\usepackage{array}
\usepackage{amsthm}
\usepackage{mathtools}
\usepackage{hyperref}
\usepackage{amsmath,bm}
\usepackage{subcaption,multirow,cleveref}
\usepackage{algorithm}
\usepackage{algpseudocode}

\usepackage{dsfont}
\crefname{algocf}{alg.}{algs.}
\Crefname{algocf}{Algorithm}{Algorithms}

\DeclareMathOperator*{\argmin}{arg\,min}
\begin{document}

\title{Maximum In-Support Return Modeling for Dynamic Recommendation with Language Model Prior}

\author{Xiaocong Chen}
\affiliation{%
  \institution{Data 61, CSIRO}
  \city{Eveleigh}
  \country{Australia}
}
\email{xiaocong.chen@data61.csiro.au}
\author{Siyu Wang}
\affiliation{%
  \institution{Macquarie University}
  \city{Sydney}
  \country{Australia}}
\email{siyu.wang@mq.edu.au}

\author{Lina Yao}
\affiliation{%
  \institution{Data 61, CSIRO}
  \city{Eveleigh}
  \country{Australia}}
\affiliation{
  \institution{The University of New South Wales}
  \city{Sydney}
  \country{Australia}}
\email{lina.yao@data61.csiro.au}


\begin{abstract}
Reinforcement Learning-based recommender systems (RLRS) offer an effective way to handle sequential recommendation tasks but often face difficulties in real-world settings, where user feedback data can be sub-optimal or sparse. In this paper, we introduce MDT4Rec, an offline RLRS framework that builds on the Decision Transformer (DT) to address two major challenges: learning from sub-optimal histories and representing complex user-item interactions. First, MDT4Rec shifts the trajectory stitching procedure from the training phase to action inference, allowing the system to shorten its historical context when necessary and thereby ignore negative or unsuccessful past experiences. Second, MDT4Rec initializes DT with a pre-trained large language model (LLM) for knowledge transfer, replaces linear embedding layers with Multi-Layer Perceptrons (MLPs) for more flexible representations, and employs Low-Rank Adaptation (LoRA) to efficiently fine-tune only a small subset of parameters. We evaluate MDT4Rec on five public datasets and in an online simulation environment, demonstrating that it outperforms existing methods. 
\end{abstract}


\maketitle

\section{Introduction}
Reinforcement Learning-based Recommender Systems (RLRS) have emerged as a powerful approach for handling sequential and session-based recommendation tasks, where the ability to adaptively learn user preferences over time is crucial~\cite{chen2023deep,chen2022generative}. By modeling the recommendation process as a sequential decision-making problem, RLRS can dynamically adjust recommendations based on user feedback, leading to personalized experiences and improved user satisfaction. These systems leverage the strengths of reinforcement learning to optimize long-term user engagement, making them particularly suitable for scenarios where user behavior evolves over time.

However, RLRS face several limitations, particularly in real-world applications where deploying and training RL models in an online environment can be impractical or costly~\cite{chen2022locality}. To address these challenges, offline RL has emerged as a promising approach~\cite{levine2020offline}, giving rise to offline RL-based recommender systems (offline RLRS)\cite{chen2023opportunities}. Offline RL, also known as data-driven RL\cite{levine2020offline}, leverages extensive pre-collected datasets of user interactions for the preliminary training of agents, eliminating the need for continuous online learning and enabling the development of RLRS agents using existing data.

Recent studies have explored the application of offline RL to enhance RLRS performance~\cite{Wang_2023,zhao2023user,chen2023opportunities}. For instance, CDT4Rec~\cite{Wang_2023} integrates the Decision Transformer (DT)\cite{chen2021decision} as its core structure and employs a causal mechanism to estimate rewards effectively. Similarly, DT4Rec\cite{zhao2023user} uses DT with a focus on capturing user retention, incorporating an efficient reward prompting method to enhance the learning process. These approaches underscore the potential of offline RL to leverage diverse and rich offline data sources, offering a practical solution to the limitations of online RL deployment. However, despite their promise, offline RLRS continue to face challenges related to learning from sub-optimal data sources and representing the complexity of user-item interactions.

In the context of learning from sub-optimal data sources,~\cite{chen2024maximum} introduces a trajectory stitching strategy, where a sub-optimal trajectory is enhanced by integrating segments from a similar but optimal trajectory, thereby generating a near-optimal sequence. While this approach has the potential to improve theoretical performance, it may not accurately capture genuine user intent in practical applications, particularly when parts of the interaction history are artificially constructed. Such synthetic transitions can introduce inconsistencies between the model’s predictions and actual user behavior. To mitigate this issue, we argue that the model should disregard unsuccessful or ineffective past experiences by dynamically adjusting the length of the historical input. By shortening the interaction history and excluding prior failures, the sequence generation model is better positioned to select actions that are more likely to produce favorable outcomes.

Another challenge in offline RLRS is representing the complexity of user-item interactions, as no existing RLRS work attempts to address it~\cite{chen2023deep}. But, one of the well-known applications of transformers - Large Language Models (LLMs) has gained significant momentum in the field of traditional RS due to their powerful knowledge-transfer and representation capabilities~\cite{du2024enhancing, he2023large, zhang2023user,kim2024large}. In sequential recommendation scenarios, LLMs have shown impressive performance in both zero-shot settings~\cite{bao2023tallrec}, pre-trained models~\cite{du2024enhancing} and user-item embedding~\cite{hu2024enhancing}, often achieving results comparable to or even exceeding those of conventional recommendation algorithms.

Given the success of LLMs, there is a natural motivation to integrate their capabilities into transformer-based sequential decision-making tasks such as offline RLRS. This integration has been explored in several recent studies on offline RL~\cite{brohan2022rt,li2022pre,reid2022can}. For example,\citet{li2022pre} proposed encoding environment states using LLMs and learning policies from the decoded states. However, their approach is limited to language-only state descriptions, restricting its applicability to more complex tasks like motion control. To address this,\citet{reid2022can} initialized DT with a pre-trained LM, allowing it to process low-level agent states and actions directly, making it more suitable for tasks such as locomotion. While this represents progress, their method does not fully exploit the capabilities of LLMs, achieving performance similar to standard DT methods and underperforming compared to offline RL approaches~\cite{shi2023unleashing}. Similarly, the potential of LLMs in RLRS remains insufficiently explored, leaving open questions about how to better utilize their capabilities to enhance RLRS performance.

In this study, we present a novel method to overcome those aforementioned two challenges, named Maximum return modeling enhanced Decision Transformer with language model prior (MDT4Rec). MDT4Rec builds upon the DT structure but with several key novel innovations: (i) Unlike EDT4Rec, we incorporate the stitching process into the action inference stage rather than the training stage. Additionally, during action inference, we introduce an optimal history length search mechanism that allows the model to disregard negative or unsuccessful past experiences; (ii) MDT4Rec uses pre-trained LLMs to initialize the weights of the DT, providing rich prior knowledge to boost the performance; (iii) To improve representation learning, we replace the standard linear embedding projections with Multi-Layer Perceptrons (MLPs), enabling a more nuanced capture of user-item interactions; (iv) To enhance training efficiency, we employ Low-Rank Adaptation (LoRA), a fine-tuning method that efficiently updates the model while maintaining the benefits of the pre-trained LLM.

Extensive experiments conducted on various real-world datasets demonstrate the effectiveness of our approach. The LLM-enhanced framework not only achieves state-of-the-art performance in dynamic recommendation tasks but also provides a new direction for future offline RLRS research, highlighting the potential of combining LLMs with offline RLRS for improved recommendation quality and system efficiency.

Our contributions can be summarized as:

\begin{itemize}
    \item We introduce MDT4rec, a new framework that integrates a large language model (LLM) as a prior and employs a specialized action inference stage stitching approach, thereby improving the offline RLRS framework.
    \item To tackle the issue of absent stitching trajectories, we develop a method that allows the model to ignore negative or unsuccessful past experiences by adjusting the trajectory length.
    \item Unlike the traditional approach that uses text-based prompt to incorporate the LLM into RS, this work incorporates an LLM into offline RLRS without instruction prompting.
    \item We evaluate the effectiveness of MDT4rec through extensive experiments on five public datasets and within an online simulation environment, showing that it outperforms existing methods.
\end{itemize}

\section{Background}
\subsection{Problem Formulation}
The recommendation problem can be conceptualized as an agent striving to achieve a specific goal through learning from user interactions, such as item recommendations and subsequent feedback. This scenario is aptly formulated as a RL problem, where the agent is trained to interact within an environment, typically described as a Markov Decision Process (MDP)~\cite{chen2023deep}.
The components of an MDP are represented as a tuple $(\mathcal{S},\mathcal{A}, \mathcal{P}, \mathcal{R}, \gamma)$ where:
\begin{itemize}
\item \textbf{State} $\mathcal{S}$: The state space, with $s_t \in \mathcal{S}$ representing the state at timestep $t$, typically containing users' previous interests, demographic information, etc.
\item \textbf{Action} $\mathcal{A}$: The action space, where $a_t \in \mathcal{A}(s_t)$. $a_t$ is the action taken given a state $s_t$, usually referring to recommended items.
\item \textbf{Transition Probability} $\mathcal{P}$: Denoted as $p(s_{t+1}|s_t, a_t) \in \mathcal{P}$, this is the probability of transitioning from $s_t$ to $s_{t+1}$ when action $a_t$ is taken.
\item \textbf{Reward} $\mathcal{R}$: $\mathcal{S} \times \mathcal{A} \to \mathbb{R}$ is the reward function, where $r(s, a)$ indicates the reward received for taking action $a$ in state $s$.
\item \textbf{Discount-rate Parameter} $\gamma$: $\gamma \in [0, 1]$ is the discount factor used to balance the value of future rewards against immediate rewards.
\end{itemize}

\emph{Offline} reinforcement learning diverges from traditional RL by exclusively utilizing pre-collected data for training, without the necessity for further online interaction~\cite{levine2020offline}. In offline RL, the agent is trained to maximize total rewards based on a static dataset of transitions $\mathcal{D}$ for learning. This dataset could comprise previously collected data or human demonstrations. Consequently, the agent in offline RL lacks the capability to explore and interact with the environment for additional data collection. The dataset $\mathcal{D}$ in an offline RL-based RS can be formally described as ${(s_t^u, a_t^u, s_{t+1}^u, r_t^u)}$, adhering to the MDP framework $(\mathcal{S},\mathcal{A}, \mathcal{P}, \mathcal{R}, \gamma)$. For each user $u$ at timestep $t$, the dataset includes the current state $s_t^u \in \mathcal{S}$, the items recommended by the agent via action $a_t^u$, the subsequent state $s_{t+1}^u$, and the user feedback $r_t^u$.
\subsection{Decision Transformer}
The Decision Transformer (DT) architecture, introduced by~\cite{chen2021decision}, reformulates offline reinforcement learning as a conditional sequence modeling task. Unlike traditional RL approaches that rely on estimating value functions or optimizing policy gradients, DT directly models the decision-making process by predicting actions based on historical interaction sequences. At each time step, the model conditions on a fixed-length window of previous states, actions, and rewards—along with the target return-to-go (RTG) to predict the next action to take.

The model is built on a Transformer backbone, originally developed for natural language processing, which uses stacked self-attention layers and residual connections to capture long-range dependencies across sequences. This structure enables DT to process complex temporal relationships within trajectories, making it well-suited for learning from offline datasets where exploration is not possible.

During training, the input is structured as a sequence of alternating triplets in the form:
\begin{align*}
    \tau = \{\cdots, s_t, R_t, a_t, \cdots\},
\end{align*}
where \( s_t \) is the state at time \( t \), \( a_t \) is the corresponding action, and \( R_t \) is the return-to-go—i.e., the expected cumulative reward from time \( t \) onward. The model learns to map these sequences to the next action in the trajectory, effectively imitating high-performing behaviors observed in the offline data.

\section{Methodology}

\begin{figure*}[ht]
    \centering
    \includegraphics[width=\linewidth]{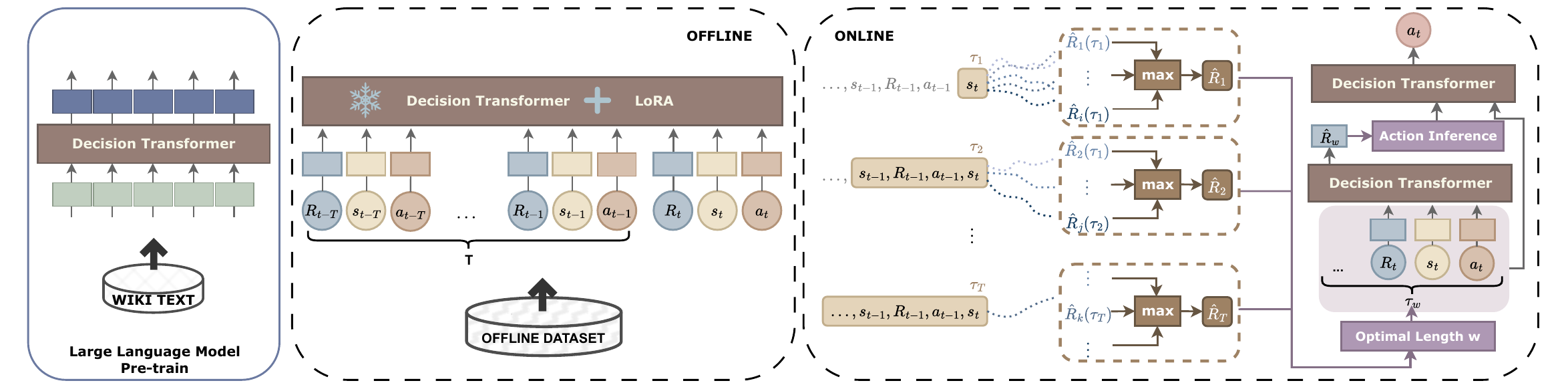}
    \caption{The overall structure of the proposed MDT4Rec}
    \label{fig:mdt}
\end{figure*}
\subsection{Maximum In-Support Return Modeling}
In contrast to EDT4Rec~\cite{chen2024maximum}, we present an alternative strategy for trajectory stitching. Rather than emphasizing the training phase, we shift this stitching process to the action inference stage. Consider a simplified scenario in which the dataset $\mathcal{D}$ contains only two trajectories, $\mathcal{D} = { (s_{t-1}^a,s_t,s_{t+1}^a),(s_{t-1}^b,s_t,s_{t+1}^b)}$. In this setup, a sequence model trained on such data will typically forecast subsequent states in alignment with the original trajectories.

To address this issue, we propose a method that enables the model to begin at \( s_{t-1}^b \) and end at \( s_{t+1}^a \) by flexibly selecting the length of historical context. Central to this approach is a maximal return estimator \( \hat{R} \), designed to assess the highest possible return over all trajectories in the dataset. This estimator guides the selection of history length that yields the maximum estimated return. For instance, if the starting point is \( s_{t-1}^b \), the model chooses to keep only \( (s_t) \) when reaching \( s_t \), since \( \hat{R}(s_t) > \hat{R}(s_{t-1}, s_t) \). In contrast, starting from \( s_{t-1}^a \), it retains the sequence \( (s_{t-1}^a, s_t) \) at \( s_t \), given that \( \hat{R}(s_{t-1}^a, s_t) \geq \hat{R}(s_t) \). This example illustrates that the most beneficial history length depends on the trajectory’s quality and may range from a single step to the maximum permitted context.

To determine the most suitable history length in a general setting, we formulate the following optimization objective:
\begin{align*}
    \arg\max_T \max_{\tau_T \in \mathcal{D}} R_t(\tau_T),
\end{align*}
where \( \tau_T \) denotes a trajectory segment of length \( T \). The structure of \( \tau_T \) is defined as:
\begin{align*}
    \tau_T = (s_{t-T+1}, R_{t-T+1}, a_{t-T+1}, \ldots, s_{t-1}, R_{t-1}, a_{t-1}, s_t, R_t, a_t).
\end{align*}

We adopt the same training framework as used in the DT, with a key modification: our approach is focused on estimating the upper bound of returns achievable given histories of different lengths. To approximate the inner maximization 
\[
\max_{\tau_T \in \mathcal{D}} R_t(\tau_T),
\]
we apply expectile regression, a method commonly used in reinforcement learning, notably in Implicit Q-Learning (IQL)~\cite{kostrikov2021offline}. Expectile regression provides a way to estimate the upper tail of the return distribution, which is particularly useful when data coverage is sparse. The \( \alpha \)-expectile of a random variable \( X \), with \( \alpha \in (0,1) \), is the solution to the following asymmetric least squares minimization:

\begin{align}
\arg\min_{m_\alpha} \mathbb{E}_{x \in X} \big[ L_2^\alpha(x - m_\alpha) \big],
\end{align}
where the asymmetric squared loss is given by \( L_2^\alpha(u) = |\alpha - \mathds{1}(u < 0)| u^2 \).

By employing expectile regression, we can approximate the term 
$\max_{\tau_T\in\mathcal{D}}R_t(\tau_T)$: 

\begin{align} 
    \hat{R}_t^T = \max_{\tau_T\in\mathcal{D}}R_t(\tau_T) \approx \argmin_{\hat{R}_t(\tau_T)} \mathbb{E}_{\tau_T\in \mathcal{D}} \bigl[L_2^\alpha \bigl(\hat{R}_t(\tau_T)-R_t\bigr)\bigr], \label{eq:rhat} 
\end{align} 

where we estimate $\hat{R}_t$ by minimizing the empirical loss of \Cref{eq:rhat} with a sufficiently high $\alpha$ value (we choose $\alpha=0.99$ in all experiments). The following theorem supports this choice:
\begin{theorem}\label{thm:expectile_convergence}
Let $ \mathbf{S}_t \coloneqq \bigl[ s_{t-T+1}, R_{t-T+1}, a_{t-T+1} \,,\, s_t \bigr].$
For $\alpha \in (0,1)$, define $\mathbf{g}^\alpha(\mathbf{S}_t) \;=\; \pi^*_{\theta}(\mathbf{S}_t),$
where $\pi^*_{\theta} \;=\; \arg \min_{m} \mathcal{L}_2^\alpha(m)$,
and
\[
\mathcal{L}_2^\alpha(m) 
\;=\; 
\mathbb{E}\Bigl[\Bigl|\alpha - \mathds{1}\bigl(R_t < m\bigr)\Bigr|\bigl(R_t - m\bigr)^2\Bigr].
\]
Assume $R^{\max}_t = \max_{\tau_T\sim \mathcal{D}} R_t(\tau_T)$ is finite. Then
\[
\lim_{\alpha \to 1} \mathbf{g}^\alpha(\mathbf{S}_t)
\;=\;
R^{\max}_t,
\]
where $R^{\max}_t$ is the maximum returns-to-go over the offline dataset.
\end{theorem}

\begin{proof}
It is a known property of $\alpha$-expectiles that $\mathbf{g}^\alpha$ is non-decreasing in $\alpha$.  
In more precise terms, if $0 < \alpha_1 < \alpha_2 < 1$, then $ \mathbf{g}^{\alpha_1} \;\le\; \mathbf{g}^{\alpha_2}.$
This arises from the structure of the piecewise weighting $\Bigl|\alpha - \mathds{1}\bigl(R_t < m\bigr)\Bigr|$,
which moves the balance point of the squared-error function to the right as $\alpha$ increases.

We claim $\mathbf{g}^\alpha \leq R^{\max}_t$ for all $\alpha \in (0,1)$.  
Suppose, for contradiction, that there exists an $\alpha_3$ with $\mathbf{g}^{\alpha_3} > R^{\max}_t$.  
By definition of $R^{\max}_t$, all returns in the dataset are strictly less than $\mathbf{g}^{\alpha_3}$.  

Consider the loss
\[
\mathcal{L}_2^{\alpha_3}(m)
\;=\;
\mathbb{E}\Bigl[\Bigl|\alpha_3 - \mathds{1}\bigl(R_t < m\bigr)\Bigr|\bigl(R_t - m\bigr)^2\Bigr].
\]
If $m = \mathbf{g}^{\alpha_3} > R^{\max}_t$, then for every return $R_t$ in the dataset, we have $R_t < m$.  
This means $\bigl(R_t - m\bigr)^2$ is strictly positive and grows larger as $m$ moves further away from $R_t$.  
Thus, decreasing $m$ toward $R^{\max}_t$ would yield a smaller squared difference and hence a smaller loss,  
contradicting the assumption that $m = \mathbf{g}^{\alpha_3}$ minimizes $\mathcal{L}_2^{\alpha_3}(m)$.  
Therefore, we must have $\mathbf{g}^{\alpha_3} \le R^{\max}_t$.  
Since $\alpha_3$ was arbitrary, it follows that $\mathbf{g}^\alpha \le R^{\max}_t$ for all $\alpha \in (0,1)$.

We have shown that $\{\mathbf{g}^\alpha\}_{\alpha \in (0,1)}$ is a non-decreasing set, and it is bounded above by $R^{\max}_t$.  
Hence, $\mathbf{g}^\alpha$ converges to some limit $\ell \le R^{\max}_t$ as $\alpha \to 1$.  
To see that $\ell = R^{\max}_t$, note that increasing $\alpha$ places more emphasis on high returns $R_t$ in the dataset,  
since the penalty for being above large values diminishes relative to the penalty for being below them.  
In the limit as $\alpha \to 1$, the minimizer of $\mathcal{L}_2^\alpha(m)$ must coincide with (or approach arbitrarily closely)  
the largest returns present in the dataset, that is $R^{\max}_t$.  Therefore, $
\lim_{\alpha \to 1} \mathbf{g}^\alpha(\mathbf{S}_t) \;=\; R^{\max}_t.$

\end{proof}

In summary, this theorem indicates that minimizing $\mathcal{L}_2^\alpha$ drives the model toward predicting the maximum returns-to-go when $\alpha$ approaches 1. This aligns with the objective of maximizing returns in RL (see \Cref{eq:rhat}). Consequently, we set $\alpha=0.99$ to avoid extensive parameter tuning.

\subsection{Action Inference }
During the action inference phase at test time, we begin by estimating the maximum achievable return, \( \hat{R}_i \), for each possible history length \( i \). Actions are then determined based on a truncated version of the observed trajectory, where truncation occurs at the history length corresponding to the highest estimated return \( \hat{R}_i \), as shown in ~\Cref{fig:mdt}.  

To identify the history length \( i \) that maximizes \( \hat{R}_i \), we employ a search strategy outlined in Algorithm~1. A brute-force search over all possible history lengths from 1 to \( T \) can be computationally inefficient, potentially leading to slow inference. To mitigate this, we incorporate a step size parameter \( \delta \) to accelerate the search procedure. This modification reduces inference time by a factor of \( \delta \), and in practice, also leads to better policy performance. We perform an ablation study to examine how different values of \( \delta \) influence the results. Following the hyperparameter evaluation in~\Cref{sec:hyper}, we use \( \delta = 2 \) across all experiments.

To sample from the expert return distribution $P(R_t,\ldots\mid \text{expert}^t)$, we adopt a similar approach to~\cite{lee2022multi} by applying Bayes’ rule: 
\begin{align} 
    P(R_t \mid \text{expert}^t, \ldots) \propto \exp(\kappa R_t) P(R_t), \label{eq:bays} 
\end{align} 
where $\kappa$ is the inverse temperature, set to $10$ as recommended in~\cite{lee2022multi}.

\begin{algorithm}
\caption{MDT4Rec optimal history length search}
\begin{algorithmic}[1]
\Require A query sequence $\tau = \langle \ldots, s_{t-1}, R_{t-1}, a_{t-1}, s_t \rangle$ and MDT4Rec model $\theta$
\For{$T = 1, 1+\delta, 1+2\delta, \ldots, T$}
    \State Obtain $\hat{R}_t(\tau_T)$ with truncated sequence 
    $\tau_T = \langle s_{t-T+1}, R_{t-T+1}, a_{t-T+1}, \ldots, s_t \rangle$ and $\theta$
\EndFor
\State Compute $P(R_t)$ with $\theta$ for $\tau_T$ that has the highest $\hat{R}_t$ and then estimate $P(R_t \mid \text{expert}^t, \ldots)$ with~\Cref{eq:bays}.
\end{algorithmic}
\end{algorithm}

When the input history is shorter, the prediction model tends to produce outputs with greater variability, which is often seen as a drawback in forecasting tasks. However, this variability can enable the sequence prediction model to explore different trajectories and potentially uncover better ones. In contrast, if the current trajectory is already close to optimal, using the longest available history helps maintain stability and reliability in predictions. This suggests a dependency between trajectory quality and the amount of historical data incorporated. Building on this insight, we propose an adaptive approach that dynamically adjusts history length to improve predictive accuracy.

\subsection{Large Language Model as Prior}
Now we present how we use the LLM as the prior to enhance the offline RLRS capability. 
The initial step involves acquiring pre-trained language models (LMs). Due to its widespread use and computational efficiency, we employ the GPT-2 architecture~\cite{radford2019language}, utilizing its readily available pre-trained weights from Hugging Face. We find that GPT-2 is strong enough to perform as the prior of the proposed MDT4Rec. The experiment that is used to support our claim is presented in the~\Cref{sec:differentllm}.

Following the method outlined in previous work~\cite{reid2022can}, this pre-training is done using the WikiText dataset~\cite{merity2016pointer}, with the standard next-token prediction objective:
\[
    \mathcal{L}_{language} = \sum_{i=1}^{s-1} -\log (T(w_{i+1}|w_1,\dots,w_i)),
\]
where \(w_i\) is the \(i\)-th token in the sequence, and \(T\) denotes the model's predicted probability distribution for the next token. We experiment with three model variations: 

\textbf{Using multi-layer perceptrons for embeddings}. In pre-trained LMs, the input is transformed into latent vectors, which are then decoded through linear projections. However, for optimal performance in offline RL, we find that replacing these linear projections with multi-layer perceptrons (MLPs) is essential. This adjustment effectively mitigates the domain mismatch between language tasks and RL tasks. Our comprehensive ablation studies, discussed in~\Cref{sec:aba}, demonstrate the critical role of this non-linear module.

\textbf{Frozen weights and low-rank adaptation}. For efficient parameter updates, we utilize LoRA~\cite{hu2021lora}, which confines gradient updates to a lower-dimensional subspace by decomposing the weight matrix \(W \in \mathbb{R}^{d \times k}\) as \(W_0 + \Delta W = W_0 + BA\), where \(B \in \mathbb{R}^{d \times r}\) and \(A \in \mathbb{R}^{r \times k}\), with \(r \ll \min(d, k)\). This method introduces low-rank matrices to the attention weights (Q, K, V), while all other Transformer weights remain fixed.

\textbf{Language prediction as an auxiliary task}. To stabilize training and preserve linguistic knowledge, we simultaneously train the model on language prediction tasks using the WikiText dataset~\cite{merity2016pointer}, akin to the pre-training phase. During language prediction, we temporarily substitute the input and output projections with those of the pre-trained LM. This auxiliary task, similar to~\cite{reid2022can}, empirically helps prevent overfitting which is validated in~\Cref{fig:language}. However, since text data may not be available for all datasets, we treat it as an auxiliary task in our work and leave the capability to process text-based data for future exploration.

\subsection{Training Objective}
The overall training objective consists of three components: DT, language modeling, and maximizing in-support estimation (i.e.,~\Cref{eq:rhat}). In this discussion, we primarily focus on the DT.  
Let \( \theta_e \) and \( \theta_g \) represent the trainable parameters of the reward estimation network \( N_e \) and the action generation network \( N_g \), respectively. Additionally, let \( \theta_s \) and \( \theta_a \) denote all trainable parameters responsible for predicting the state sequence \( s_p^{t-T+1:t} \) and action sequence \( a_p^{t-T+1:t} \).  The reward estimation network \( N_e \), along with the state and action prediction components, are trained by minimizing the squared error between the estimated and actual rewards, leading to the following loss function \( \mathcal{L}_{N_e}(\theta_e, \theta_s, \theta_a) \):
\begin{align*}
\mathbb{E}_{(R,s,a)\sim \tau}\Bigg[\frac{1}{T}\sum_{t=1}^{T} \bigg( r_t - N_e\big(s_p^{t-T+1:t}(\theta_s), a_p^{t-T+1:t}(\theta_a);\theta_e\big) \bigg)^2 \Bigg]
\end{align*}

For the action generation network \( N_g \), the objective is to optimize action selection by minimizing the cross-entropy loss, defined as \( \mathcal{L}_{N_g}(\theta_g, \theta_e, \theta_a) \):

\begin{align*}
    \frac{1}{T}\mathbb{E}_{(R,s,a)\sim \tau} \Bigg[ -\sum_{t=1}^{T} \log N_g\big(\hat{r}^{t-T+1:t}(\theta_e), a_p^{t-T+1:t}(\theta_a);\theta_g \big) \Bigg]
\end{align*}

Thus, the overall training objective for the DT model is formulated as $ \mathcal{L}_{dt} =  \mathcal{L}_{N_e}(\theta_e, \theta_s, \theta_a) + \mathcal{L}_{N_g}(\theta_g, \theta_e, \theta_a)$. Consequently, the overall training objective for MDT4Rec becomes:
\begin{align*}
    \mathcal{L} = \mathcal{L}_{dt} + \lambda\cdot \mathcal{L}_{language} + \mathcal{L}_\text{max},
\end{align*}
where $\lambda$ is set to be in $0.1$ to stabilized the training process. However, whether this term could enhance performance in all cases still requires further investigation.

\section{Experiments}
\begin{table*}[!ht]
\centering
\caption{Performance comparison between selected RL methods trained with a user model on five datasets and online simulation environment. The best results are highlighted in bold. SASRec is used as the state encoder for those offline datasets and GPT2 as the LLM.}
\resizebox{\textwidth}{!}{
\begin{tabular}{c|ccc|ccc|ccc}
\hline

\textbf{Method} & \multicolumn{3}{c|}{\textbf{Coat}} & \multicolumn{3}{c|}{\textbf{MoveLens}} & \multicolumn{3}{c}{\textbf{KuaiRec}} \\
 & $R_{cumu}$ & $R_{avg}$ & Length & $R_{cumu}$ & $R_{avg}$ & Length & $R_{cumu}$ & $R_{avg}$ & Length \\
\hline
DDPG & $15.2764 \pm 7.02$ & $2.1186 \pm 0.97$ & 7.21 & $9.0271 \pm 4.32$ & $2.4721 \pm 1.18$ & 3.65 & $9.4018 \pm 4.24$ & $1.0197 \pm 0.46$ & 9.22 \\
SAC & $80.9823 \pm 24.22$ & $2.9008 \pm 0.78$ & 27.89 & $34.1192 \pm 8.92$ & $3.4624 \pm 1.12$ & 9.85 &  $23.4216 \pm 10.02$&  $0.7473 \pm 0.32$   & 31.34 \\
TD3 & $16.0280 \pm 7.48$ & $2.2832 \pm 1.06$ & 7.02 & $9.7262 \pm 4.66$ & $2.5791 \pm 1.24$ & 3.77 & $8.2458 \pm 3.89$ & $0.9154 \pm 0.43$ & 9.01 \\
DT & $83.4891 \pm 25.88$ & $2.7830 \pm 0.86$ & 30.02 & $33.2188 \pm 9.02$ & $3.4884 \pm 1.05$ & 9.52 & $29.0291 \pm 10.00$ & $0.9080 \pm 0.31$ & 31.98 \\
DT4Rec & $82.1823 \pm 26.22$ & $2.6156 \pm 0.83$ & 31.42 & $32.8728 \pm 8.91$ & $3.3238 \pm 0.91$ & 9.89 & $28.5418 \pm 9.51$ & $0.8798 \pm 0.29$ & 32.44\\
CDT4Rec & $87.8712 \pm 24.54$ & $2.9010 \pm 0.80$ & 30.29 & $35.1284 \pm 9.88$ & $4.0010 \pm 1.13$ & 8.78 & $30.4888 \pm 11.01$ & $1.0061 \pm 0.36$ & 30.22 \\ 
EDT4Rec & $88.1922 \pm 22.92$ & $2.8972 \pm 0.76$ & 30.42 & $35.8312 \pm 10.13$ & $3.9777 \pm 1.14$ & 8.89 & $31.0726 \pm 11.66$ & $1.0397 \pm 0.39$ & 29.88\\
\hline
\textbf{MDT4Rec} & $\mathbf{89.8742 \pm 22.48}$ & $\mathbf{2.9275 \pm 0.73}$ & 30.70 & $\mathbf{38.6401 \pm 10.71}$ & $\mathbf{4.2933 \pm 1.20}$ & 9.00 & $\mathbf{33.0124 \pm 11.90}$ & $\mathbf{1.0726 \pm 0.39}$ & 30.78 \\
\textbf{MDT4Rec-LM} & $88.7680 \pm 22.80$ & $2.9095 \pm 0.75$ & 30.51 & $37.6233 \pm 10.90$ & $3.9855 \pm 1.15$ & 9.44 & $32.0234 \pm 11.70$ & $1.0604 \pm 0.39$ & 30.20 \\
\textbf{MDT4Rec-Max} & $87.9283 \pm 23.24$ & $2.9029 \pm 0.77$ & 30.29 & $36.3421 \pm 10.02$ & $4.0425 \pm 1.11$ & 8.99 & $31.8272 \pm 11.60$ & $1.0501 \pm 0.38$ & 30.31 \\
\hline
\end{tabular}
}

\resizebox{\textwidth}{!}{
\begin{tabular}{c|ccc|ccc|ccc}

\hline

\textbf{Method} & \multicolumn{3}{c|}{\textbf{YahooR3}} & \multicolumn{3}{c|}{\textbf{KuaiRand}} & \multicolumn{3}{c}{\textbf{VirtualTB}} \\
 & $R_{cumu}$ & $R_{avg}$ & Length & $R_{cumu}$ & $R_{avg}$ & Length & $R_{cumu}$ & $R_{avg}$ & Length \\
\hline
DDPG & $7.4421 \pm 3.12$ & $2.9063 \pm 1.22$ & 2.56 & $1.3972 \pm 0.52$ & $0.3288 \pm 0.12$ & 4.25 & $74.6534 \pm 24.22$ & $5.6474 \pm 1.83$ & 13.22 \\
SAC & $73.6721 \pm 21.01$&    $2.7748 \pm 0.79$   & 26.55 &  $6.8671 \pm 1.47$  &  $0.6590 \pm 0.14$    & 10.42 & $73.4271 \pm 21.92$&   $5.5125 \pm 1.65$   & 13.32\\
TD3 & $8.0623 \pm 3.35$ & $3.1243 \pm 1.30$ & 2.58 & $1.4765 \pm 0.47$ & $0.2828 \pm 0.09$ & 5.22 & $75.2452 \pm 22.42$ & $6.0442 \pm 1.80$ & 12.45 \\
DT & $79.8234 \pm 22.55$ & $2.8512 \pm 0.80$ & 28.04 & $6.9631 \pm 1.95$ & $0.6312 \pm 0.18$ & 11.03 & $78.8762 \pm 23.99$ & $5.9753 \pm 1.83$ & 13.20 \\
DT4Rec & $78.5172 \pm 23.02$ & $2.6688 \pm 0.78$ & 29.42 & $6.8172 \pm 1.82$ & $0.5686 \pm 0.15$ & 11.99 & $76.7871 \pm 22.47$ & $5.4420 \pm 1.59$ & 14.11\\
CDT4Rec & $80.7552 \pm 21.89$ & $3.2150 \pm 0.86$ & 25.41 & $7.3271 \pm 1.78$ & $0.6508 \pm 0.16$ & 11.27 & $79.2101 \pm 22.31$ & $5.6498 \pm 1.59$ & 14.02 \\ 
EDT4Rec & $82.4211 \pm 21.97$ & $3.4385 \pm 0.92$ & 23.97 & $7.5817 \pm 1.78$ & $0.6497 \pm 0.15$ & 11.67 & $79.6651 \pm 20.44$ & $5.6741 \pm 1.46$ & 14.04\\
\hline
\textbf{MDT4Rec} & $\mathbf{83.1846 \pm 23.20}$ & $\mathbf{3.6436 \pm 1.01}$ & 23.02 & $\mathbf{7.7128 \pm 1.62}$ & $\mathbf{0.6614 \pm 0.14}$ & 11.67 & $\mathbf{81.7891 \pm 21.78}$ & $\mathbf{5.9831 \pm 1.59}$ & 13.67 \\
\textbf{MDT4Rec-LM} & $82.4655 \pm 20.88$ & $3.4375 \pm 0.87$ & 23.99 & $7.6322 \pm 1.70$ & $0.6523 \pm 0.15$ & 11.70 & $80.7261 \pm 20.98$ & $5.8160 \pm 1.52$ & 13.88 \\
\textbf{MDT4Rec-Max} & $82.4012 \pm 21.01$ & $3.4320 \pm 0.88$ & 24.01 & $7.5662 \pm 1.70$ & $0.6749 \pm 0.15$ & 11.65 & $79.8721 \pm 21.01$ & $5.7133 \pm 1.50$ & 13.98\\
\hline
\end{tabular}
}
\label{tab:allresult_2}
\end{table*}

In this section, we present experimental results aimed at addressing the following four key research questions:
\begin{itemize}
    \item \textbf{RQ1}: How does MDT4Rec perform relative to existing offline RL-based recommendation approaches and standard reinforcement learning algorithms in both online and offline recommendation settings?
    \item \textbf{RQ2}: What is the impact of using different large language models (LLMs) on the performance of MDT4Rec?
    \item \textbf{RQ3}: What roles do the LoRA modules and the modified MLP layers play in shaping the final performance? What is the contribution of each individual component in the MDT4Rec framework?
    \item \textbf{RQ4}: How do variations in history length \( T \) and step size \( \delta \) influence the model's overall effectiveness?
\end{itemize}

\subsection{Datasets and Environments}
In this section, we evaluate the performance of our proposed algorithm, MDT4Rec, against other state-of-the-art algorithms, employing both real-world datasets and an online simulation environment. We first introduce five diverse, public real-world datasets from various recommendation domains for our offline experiments:
\begin{itemize}
    \item Coat~\cite{schnabel2016recommendations}, is used for product recommendation with 290 users, 300 items, and 7,000 interactions for training, and 290 users, 300 items, and 4,600 interactions for testing.

\item YahooR3~\cite{marlin2009collaborative}, is used for music recommendation with 15,400 users, 1,000 items, and 311,700 interactions for training, and 5,400 users and 54,000 interactions for testing.

\item  MovieLens-1M \footnote{https://grouplens.org/datasets/movielens/1m}: It contains 58,000 users, 3,000 items, and 5,900,000 interactions for training, and 58,000 users and 4,400,000 interactions for testing.

\item KuaiRec~\cite{gao2022kuairec}, is used for video recommendation with 1,000 users, 6,700 items, and 594,400 interactions for training, and 1,000 users and 371,000 interactions for testing.

\item KuaiRand~\cite{gao2022kuairand}, is similar to KuaiRec, but with 571,300 interactions for training, and 1,000 users and 371,000 interactions for testing.
\end{itemize}
All of those mentioned datasets are converted to the RL environment via EasyRL4Rec~\cite{yu2024easyrl4rec}. The overall environment follows the OpenAI Gymnasium to ensure the formats are standardized. We summarise three key components used in this environment for reproducibility. 
\begin{itemize}
    \item State:  A StateTracker module is utilized to create state representations. This module encodes user states by considering historical actions and corresponding feedback. Popular sequential modeling techniques, such as RNNs (GRU), CNN-based models (Caser), and attention-based models (SASRec), are implemented to transform user interaction sequences into state embeddings. 
    \item Action: The action space corresponds to the items that can be recommended. If the continuous action is used, a mechanism is implemented to convert continuous action outputs into discrete item recommendations.
    \item Reward: The rewards are derived directly from user feedback available in the dataset. Depending on the dataset, the reward can be a click, rating, dwell time, or other behavior metrics that represent user engagement.
\end{itemize}
Matrix factorization is used to make predictions for missing interactions. Evaluations are conducted on 100 parallel test environments.

Furthermore, we evaluate our proposed approach using a real-world online simulation platform to assess its effectiveness in an interactive setting. Specifically, we adopt VirtualTB~\cite{shi2019virtual} as the primary testbed for online evaluation. The evaluation metrics used are consistent with those employed in the offline experiments, namely: \( R_{\text{cumu}} \), \( R_{\text{avg}} \), and \textit{Length}. Here, \( R_{\text{cumu}} \) denotes the total accumulated reward over a trajectory, while \( R_{\text{avg}} \) refers to the mean reward per interaction. The \textit{Length} corresponds to the total number of interactions within a single trajectory.

It is worth mentioning that, those offline datasets are converted to RL environments via the EasyRL4Rec~\cite{yu2024easyrl4rec} framework which has a different setting than VirtualTB. Hence, the evaluation scale might be different.

\subsection{Baselines and Experimental Setup}
In our study, we evaluate MDT4Rec within the context of offline RLRS methods. Given that our work is centered on offline RLRS, we have selected the most recent offline RLRS approaches as baselines. Additionally, we have included several well-established RL algorithms to ensure a comprehensive evaluation.
\begin{itemize}
    \item \textbf{Deep Deterministic Policy Gradient (DDPG)}~\cite{lillicrap2015continuous}: An off-policy method designed for environments with continuous action spaces.
    \item \textbf{Soft Actor-Critic (SAC)}~\cite{haarnoja2018soft}: This approach is an off-policy, maximum entropy Deep RL method that focuses on optimizing a stochastic policy.
    \item \textbf{Twin Delayed DDPG (TD3)}~\cite{fujimoto2018addressing}: An enhancement of the baseline DDPG, TD3 improves performance by learning two Q-functions, updating the policy less frequently.
    \item \textbf{DT}~\cite{chen2021decision}: An offline RL algorithm that leverages the transformer architecture to infer actions.
    \item \textbf{DT4Rec}~\cite{zhao2023user}: Building on the standard DT framework, DT4Rec integrates a conservative learning method to better understand users' intentions in offline RL4RS.
    \item \textbf{CDT4Rec}~\cite{Wang_2023}: This model introduces a causal layer to the DT framework, aiming to more effectively capture user intentions and preferences in offline RL4RS.
    \item \textbf{EDT4Rec}~\cite{chen2024maximum}: This model is an extension of the CDT4Rec which considers the stitching problem in DT.
\end{itemize}
It is worth noting that offline RL-based methods that require expert demonstrations for batch training. For the five offline datasets, we use the algorithm that achieves the best performance to serve as the expert policy collector. In contrast, for VirtualTB, where training data access is restricted, a separate RL algorithm is employed to explore and collect trajectories. Following prior works~\cite{Wang_2023,chen2024maximum}, we use DDPG as the behavior policy to gather a certain number of expert trajectories. Although other algorithms are also suitable for serving as the trajectories collector, DDPG is chosen due to its popularity and widespread use in the RLRS doamin~\cite{chen2023deep}. The number of trajectories is intentionally limited to provide the model with a foundational understanding of the environment without extensive exploration. 
Experiments are conducted on a server with two Intel Xeon CPU E5-2697 v2 CPUs with 6 NVIDIA TITAN X Pascal GPUs, 2 NVIDIA TITAN RTX, and 768 GB memory.

\subsection{Overall Comparison (RQ1)}

The results in Table~\ref{tab:allresult_2} show a consistent trend: \textsc{MDT4Rec} achieves the highest cumulative return ($R_{\text{cumu}}$) and per-step reward ($R_{\text{avg}}$) across all datasets. It outperforms both traditional actor–critic methods (DDPG, SAC, TD3) and recent transformer-based offline recommendation baselines (DT, DT4Rec, CDT4Rec, EDT4Rec). This strong performance holds across a diverse set of environments—including explicit-feedback datasets like Coat, implicit-feedback datasets such as KuaiRec and KuaiRand, and large-scale interaction logs like VirtualTB and YahooR3—demonstrating that \textsc{MDT4Rec} generalises well to different catalogue sizes, sparsity levels, and feedback modalities.

The improvement in average reward is particularly relevant for recommender systems, where high per-step reward reflects the model's ability to consistently select user-aligned items, rather than relying on longer trajectories to accumulate reward. The reported trajectory lengths for \textsc{MDT4Rec} remain similar to those of the strongest baselines, indicating that its reward advantage is not due to prolonged interactions but to better decision-making at each step. In practice, this suggests that users would receive more relevant content throughout the session, which is crucial for maintaining engagement.

A key observation is that \textsc{MDT4Rec}'s performance remains stable when transitioning from offline evaluation to online simulation via VirtualTB. Models that perform well on static logs often deteriorate in interactive environments due to exposure to out-of-distribution states. In contrast, \textsc{MDT4Rec} maintains its advantage, indicating that its learned policy is robust and adapts well during online interaction. This stability is likely a result of the architectural design, which combines structured sequence modelling with enhanced representation capacity.

Moreover, \textsc{MDT4Rec} exhibits lower variance across training runs compared to the actor–critic methods. This suggests that its training process is more stable and less sensitive to initialisation or sampling noise, reducing the need for extensive hyperparameter tuning and making the model more reliable in practical deployment.

\begin{table*}[!ht]
\caption{Different LLMs across five offline datasets and online simulator}
\resizebox{\textwidth}{!}{
\begin{tabular}{c|ccc|ccc|ccc}
\hline
\textbf{Method} & \multicolumn{3}{c|}{\textbf{Coat}} & \multicolumn{3}{c|}{\textbf{MoveLens}} & \multicolumn{3}{c}{\textbf{KuaiRec}} \\
 & $R_{cumu}$ & $R_{avg}$ & Length & $R_{cumu}$ & $R_{avg}$ & Length & $R_{cumu}$ & $R_{avg}$ & Length \\
\hline
\hline
 Ours (GPT2 - Large) & $89.8742 \pm 22.48$ & $2.9275 \pm 0.73$ & 30.70 & $38.6401 \pm 10.71$ & $4.2933 \pm 1.20$ & 9.00 & $33.0124 \pm 11.90$ & $1.0726 \pm 0.39$ & 30.78\\
Ours (LLaMA 3.2- 3B) & $89.7278 \pm 21.49$ & $2.8963 \pm 0.69$ & 30.98 & $39.4049 \pm 10.23$ & $4.2878 \pm 1.11$ & 9.19 & $33.2441 \pm 11.67$ & $1.0717 \pm 0.38$ & 31.02 \\
Ours (Qwen 2.5 - 7B) & $90.5261 \pm 22.04$ & $2.9373 \pm 0.72$& 30.82 & $39.7698 \pm 10.34$& $4.2948 \pm 1.12$ & 9.26 & $33.7589 \pm 11.70$ & $1.0741 \pm 0.37$ & 31.43\\
Ours (DeepSeek - R1 - 7B) & $89.8570 \pm 21.77$ & $2.9298 \pm 0.71 $ & 30.67 & $30.9865 \pm 10.66$ & $4.2950 \pm 1.15$ & 9.31 & $33.5084 \pm 11.66$ & $1.0733 \pm 0.37$ & 31.22 \\
\hline
\end{tabular}
}

\resizebox{\textwidth}{!}{
\begin{tabular}{c|ccc|ccc|ccc}

\hline
\textbf{Method} & \multicolumn{3}{c|}{\textbf{YahooR3}} & \multicolumn{3}{c|}{\textbf{KuaiRand}} & \multicolumn{3}{c}{\textbf{VirtualTB}} \\
 & $R_{cumu}$ & $R_{avg}$ & Length & $R_{cumu}$ & $R_{avg}$ & Length & $R_{cumu}$ & $R_{avg}$ & Length \\
\hline
Ours (GPT2 - Large) & $83.1846 \pm 23.20$ & $3.6436 \pm 1.01$ & 23.02 & $7.7128 \pm 1.62$ & $0.6614 \pm 0.14$ & 11.67 & $81.7891 \pm 21.78$ & $5.9831 \pm 1.59$ & 13.67 \\
Ours (LLaMA 3.2- 3B) & $90.7752 \pm 23.34$ & $3.6151 \pm 0.93$  & 25.11 & $7.1884 \pm 1.57$ & $0.6529\pm 0.14$ & 11.02 & $80.9560 \pm 20.88$  & $5.9923 \pm 1.55$ & 13.51\\
Ours (Qwen 2.5- 7B)& $93.5563 \pm 23.80$ & $3.6574 \pm 0.93$ & 25.58 & $7.6805 \pm 1.49$ & $0.6644 \pm 0.13$& 11.56 & $80.5481 \pm 21.03$ & $6.0021 \pm 1.57$ & 13.42\\
Ours (DeepSeek - R1 - 7B) & $91.1495 \pm 23.02$ & $3.6489 \pm 0.92 $ & 24.98 & $7.5635 \pm 1.44$& $0.6623 \pm 0.12$& 11.42 & $83.2638 \pm 21.43$ & $5.9902 \pm 1.54$ & 13.90\\
\hline
\end{tabular}
}
\label{tab:differentllm}
\end{table*}
\subsection{Different LLMs for Weight Initialization (RQ2)}

\label{sec:differentllm}
Table \ref{tab:differentllm} shows how four open-weight language models influence MDT4Rec when their transformer blocks serve as the initial policy backbone. In every offline dataset and in the VirtualTB simulator the overall ranking stays the same: the system trained from each LLM reaches a high cumulative return and a high per-step reward, which implies that the subsequent fine-tuning on recommendation trajectories adapts each backbone well to user–item sequences.

Some qualitative differences still emerge. The two seven-billion-parameter models, Qwen 2.5 and DeepSeek-R1, yield the strongest scores on most datasets. Their larger capacity likely supplies richer contextual embeddings at the start of fine-tuning, leading to slightly sharper state and action representations after training. 

DeepSeek-R1 performs especially well in the online simulator. Unlike offline logs, the simulator produces novel states that differ from the training data. DeepSeek-R1’s conversational pre-training may help it form more flexible latent features, allowing the policy to handle unseen interaction patterns with less drift.

GPT-2 Large and LLaMA 3.2-3B remain close behind the larger backbones. Because GPT-2 Large has a smaller footprint, it runs faster during fine-tuning and inference, which can be attractive in practical deployments where memory or latency are tight constraints. LLaMA 3.2-3B shows a similar trade-off between resource demand and final reward.

In summary, MDT4Rec benefits from larger or domain-aligned language models, yet its performance does not collapse when a lighter or more general backbone is used. This robustness lets practitioners pick an LLM that matches licence terms and hardware limits while retaining strong recommendation quality.

\subsection{Fine-Tuning Strategies and Embedding (RQ3)}

Recent works~\cite{hansen2022pre,ze2023visual} show that full finetuning representations for visual RL tasks is better than adopting the frozen pre-trained models. While some works~\cite{biderman2024lora,sun2023comparative} show that LoRA could outperform frozen and fully finetuned models.  In this section, we investigate the effectiveness of two key architectural and training choices in \textsc{MDT4Rec}: (i) the fine-tuning strategy for the pre-trained language model (LM) backbone, and (ii) the design of the input embedding layer. To ensure computational efficiency, we perform all experiments on the VirtualTB online simulation platform. GPT-2 is used as the language model for initialization in all variants.

We first examine three adaptation strategies for the LM backbone: full finetuning, frozen parameters, and low-rank adaptation via LoRA. The results are summarized in Table~\ref{tab:lora}.

\begin{table}[h]
    \centering
    \caption{Comparing LoRA with full finetuning and frozen parameters on VirtualTB.}
    \begin{tabular}{c|ccc}
        \hline
         \textbf{Method} &  $R_{\text{cumu}}$ & $R_{\text{avg}}$ & \textbf{Length} \\ \hline
         LoRA & $81.7891 \pm 21.78$ & $5.9831 \pm 1.59$ & 13.67 \\
         Frozen & $80.1210 \pm 21.44$ & $5.9305 \pm 1.73$ & 13.51\\
         Full finetuning & $77.8939 \pm 21.10 $ & $5.8435 \pm 1.66$ & 13.33\\ \hline 
    \end{tabular}
    \label{tab:lora}
\end{table}

LoRA achieves the highest cumulative reward ($R_{\text{cumu}}$) and average reward ($R_{\text{avg}}$), while maintaining a similar trajectory length to the other two variants. This indicates that LoRA enables the model to improve the quality of recommendations without relying on extended user sessions. Although the differences across methods are not large, the consistent edge of LoRA suggests that low-rank adaptation offers a practical balance between flexibility and stability. Unlike full finetuning, which may introduce instability due to overfitting or catastrophic forgetting, and unlike the frozen backbone, which restricts adaptability, LoRA allows for targeted updates that preserve useful pre-trained structure while supporting task-specific refinement.

Next, we study the impact of the embedding module used to project raw state features into the transformer’s hidden space. In \textsc{MDT4Rec}, we use a two-layer multilayer perceptron (MLP) by default. To assess its contribution, we compare it with a simpler linear projection. The results are shown in Table~\ref{tab:mlp}.

The MLP embedding clearly outperforms the linear projection on both reward metrics. Since the state representations in recommender systems often involve heterogeneous signals—such as discrete IDs, numerical statistics, and textual embeddings—a simple linear layer may be insufficient to model the complex dependencies among these inputs. The MLP introduces non-linearities that help capture higher-order interactions, enabling the downstream transformer to make more accurate predictions. The small increase in trajectory length with MLP is not large enough to account for the reward difference, suggesting the gain primarily comes from improved representation quality.

These findings highlight the contribution of both components. LoRA provides an efficient mechanism to adapt large pre-trained models without full parameter updates, and the MLP embedding layer enhances the expressiveness of state encoding. Removing either component leads to a decline in performance, showing that each plays a distinct and complementary role in enabling \textsc{MDT4Rec} to produce high-quality recommendations in interactive environments.

\begin{table}[h]
    \centering
    \caption{Comparing MLP embedding with linear projection on VirtualTB.}
    \begin{tabular}{c|ccc}
        \hline
         \textbf{Method} &  $R_{\text{cumu}}$ & $R_{\text{avg}}$ & \textbf{Length} \\ \hline
         MLP & $81.7891 \pm 21.78$ & $5.9831 \pm 1.59$ & 13.67\\
         Linear & $76.4639 \pm 23.41 $ &  $5.8728 \pm 1.80$ & 13.02 \\ \hline
    \end{tabular}
    \label{tab:mlp}
\end{table}
\subsection{Ablation Study (RQ3)}
\label{sec:aba}
Table~\ref{tab:allresult_2} provides a detailed comparison between \textsc{MDT4Rec} and two ablated variants: \textsc{MDT4Rec-LM}, which removes the language model (LM) prior, and \textsc{MDT4Rec-Max}, which removes the maximum in-support return modeling.

Comparing \textsc{MDT4Rec} with \textsc{MDT4Rec-LM} shows a slight reduction in both cumulative return ($R_{\text{cumu}}$) and average reward ($R_{\text{avg}}$) when the LM prior is removed. This indicates that initializing with a pretrained language model supports better generalization and knowledge transfer, likely by providing richer contextual representations for user-item sequences and historical behaviors. As a result, the policy becomes more adaptive and consistently achieves higher rewards.

We further examine the role of language-level supervision by incorporating a weighted language modeling loss during training. As shown in Figure~\ref{fig:language}, this auxiliary loss helps prevent significant performance degradation during the reinforcement learning stage and reduces training variance. These findings suggest that language prediction acts as an effective regularization mechanism, helping to preserve the representational strength of the pretrained model while facilitating learning in a data-limited decision-making context.

The performance gap between \textsc{MDT4Rec} and \textsc{MDT4Rec-Max} highlights the importance of dynamically selecting effective history lengths. Without maximum in-support return modeling, the model lacks the ability to identify and discard low-reward subsequences during training, leading to weaker outcomes. In contrast, \textsc{MDT4Rec} can adaptively retain high-quality trajectory segments and suppress noisy or uninformative ones, which is especially beneficial when training on sub-optimal offline data.

\textsc{MDT4Rec} integrates both components—the LM prior and maximum in-support return modeling—and consistently achieves higher $R_{\text{cumu}}$ and $R_{\text{avg}}$ across all evaluated datasets. This combination suggests that leveraging pretrained knowledge and trajectory-level selection jointly strengthens the policy and improves robustness under sparse or noisy feedback conditions.
\begin{figure*}[!ht]
     \centering
     \begin{subfigure}[b]{0.33\linewidth}
         \centering
           \includegraphics[width=\linewidth]{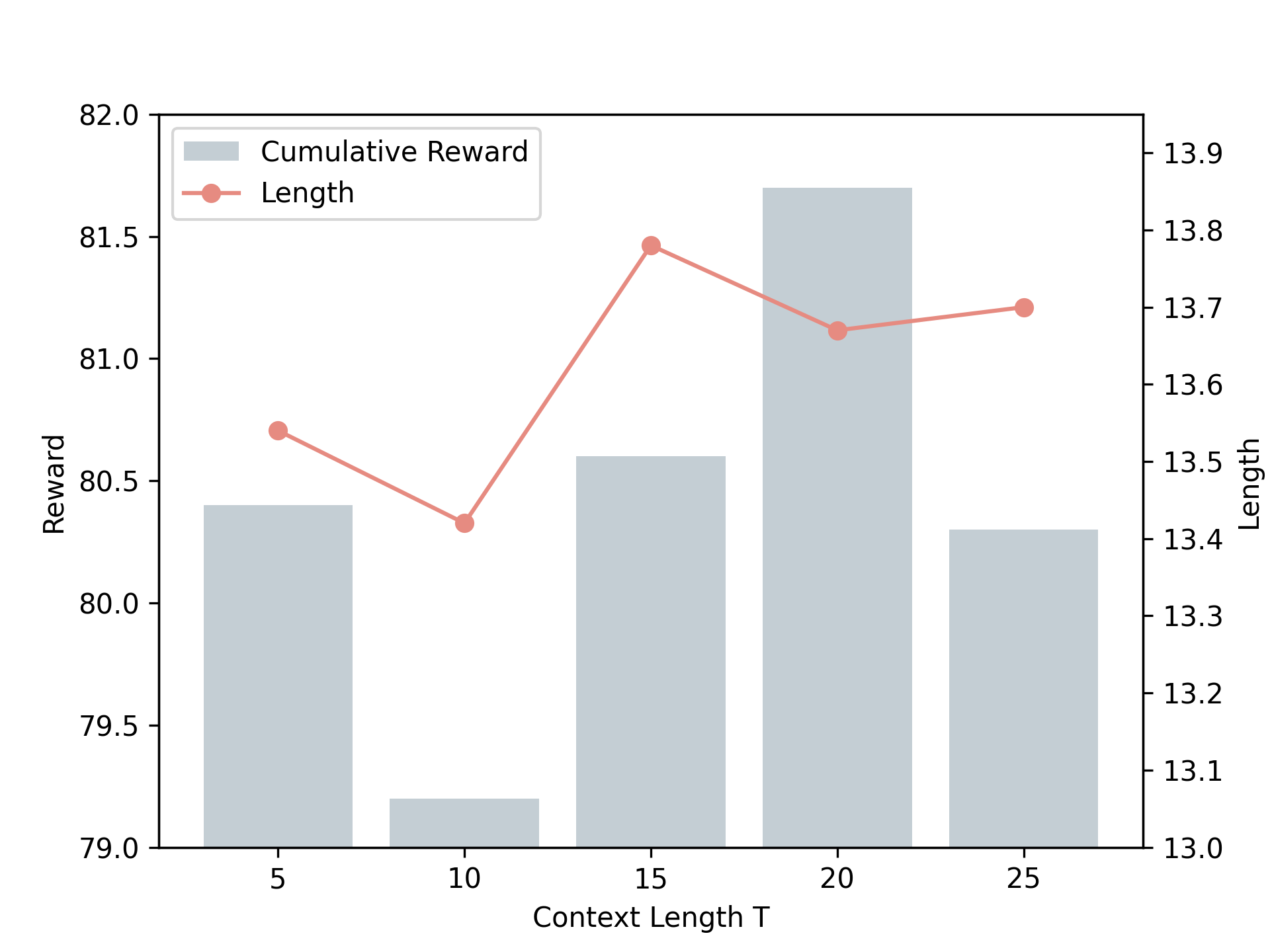}
           \caption{}
           \label{fig:hyper}
     \end{subfigure}
     \begin{subfigure}[b]{0.33\linewidth}
         \centering
         \includegraphics[width=\linewidth]{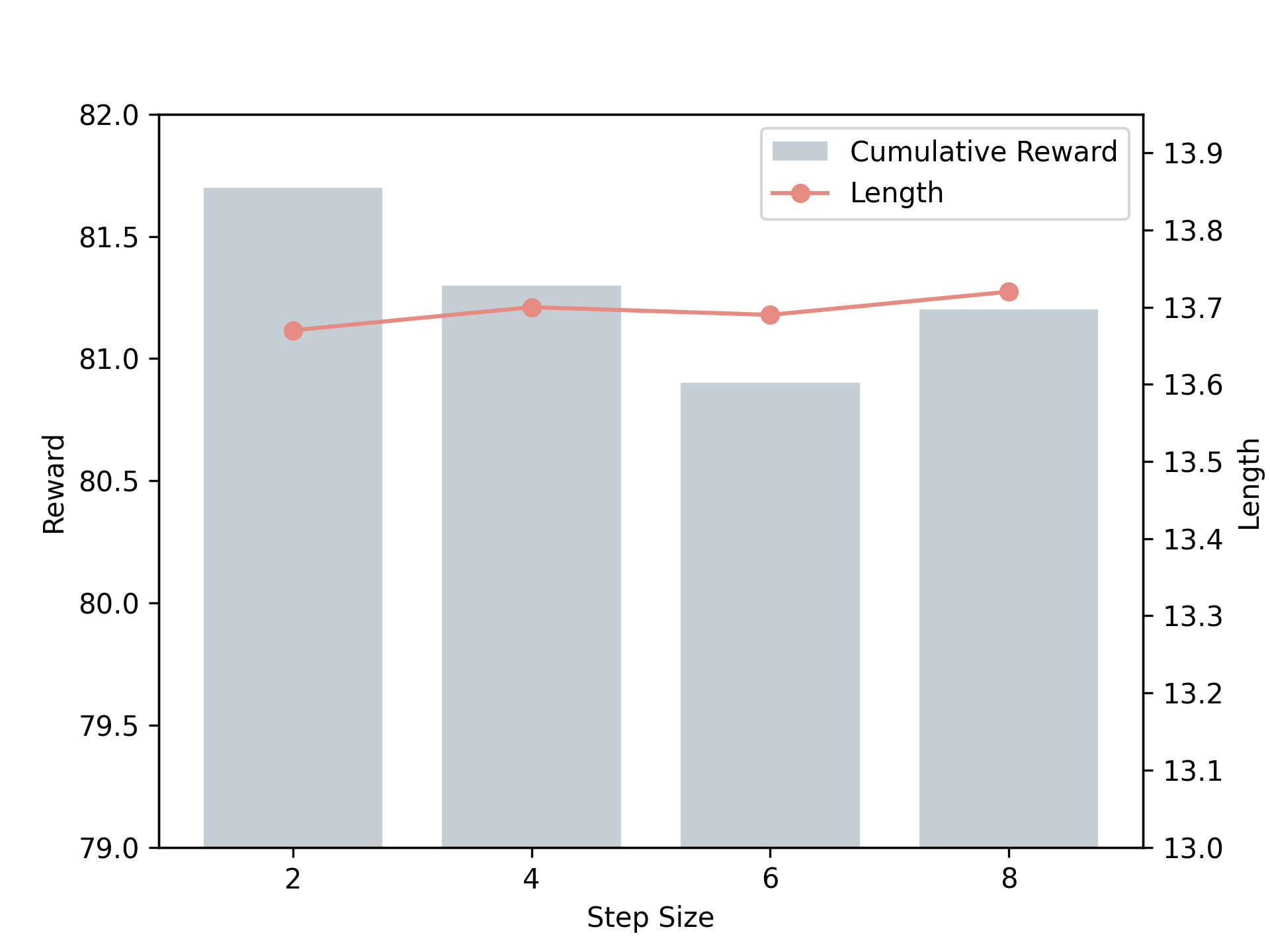}
         \caption{}
         \label{fig:delta}
     \end{subfigure}
     \begin{subfigure}[b]{0.33\linewidth}
         \centering
         \includegraphics[width=\linewidth]{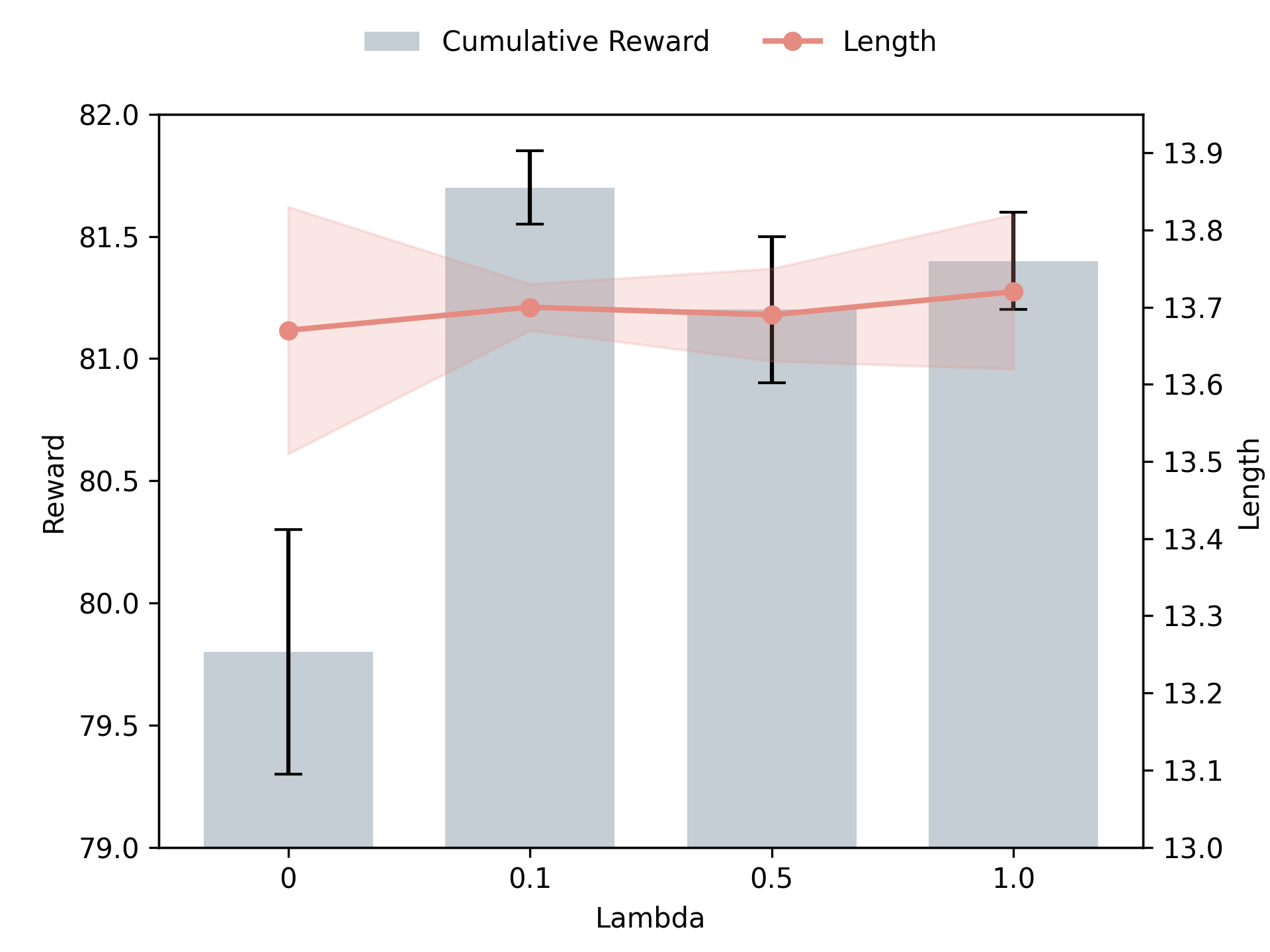}
         \caption{}
         \label{fig:language}
     \end{subfigure}
        \caption{(a). History Length study. (b). Step Size $\delta$ Study. (c). Language Auxiliary Loss Study}
\end{figure*}
\subsection{Hyperparameter Study (RQ4)}
\label{sec:hyper}
In this section, we analyze the effects of key hyperparameters in \textsc{MDT4Rec}, focusing on the context length \( T \) and the step size \( \delta \).

We first examine the impact of context length \( T \), selected from the range \(\{5, 10, 15, 20, 25\}\). The results are presented in Figure~\ref{fig:hyper}. Since \( R_{\text{avg}} \) is computed as the ratio of \( R_{\text{cumu}} \) to sequence length, we report only \( R_{\text{cumu}} \) and the trajectory length for clarity. As shown in the figure, \textsc{MDT4Rec} achieves the highest cumulative reward at \( T = 20 \), while maintaining a relatively short sequence length. This combination indicates that \( T = 20 \) also yields the highest average reward, making it the most effective context length in terms of both return and efficiency. We therefore adopt \( T = 20 \) as the default setting in our experiments.

Next, we study the effect of step size \( \delta \), with values selected from \(\{2, 4, 6, 8\}\), while fixing \( T = 20 \). The corresponding results are shown in Figure~\ref{fig:delta}. We observe that \( R_{\text{cumu}} \) fluctuates across different step sizes: it peaks at \( \delta = 2 \), drops at \( \delta = 6 \), and slightly recovers at \( \delta = 8 \). The sequence length remains relatively stable across configurations, with a slight increase at larger step sizes. Considering both cumulative reward and sequence length, the highest average reward is obtained when \( \delta = 2 \), indicating this setting best balances return quality and efficiency.

\section{Related Work}
\subsection{Offline Reinforcement Learning based Recommendation}
Recent studies have started exploring the prospect of integrating offline RLRS~\cite{chen2023opportunities}. Notably,~\citet{Wang_2023,wang2024retentive} introduced a novel model called the Causal Decision Transformer for RS. This model incorporates a causal mechanism designed to estimate the reward function, offering promising insights into the offline RL-RS synergy. As an extension,~\citet{chen2024maximum} further explore the stitching capability of DT when dealing with the sub-optimal trajectories. Similarly,~\citet{zhao2023user} introduced the DT4Rec, utilizing the vanilla Decision Transformer as its core architecture to provide recommendations, demonstrating its potential in the field.~\cite{wang2025policy} propose a causal approach to enhance the offline RLRS by decomposing the state space.~\cite{zhang2025darlr} design a reward model method for the Model-based offline RLRS by incroproting the world model.
\subsection{Large Language Model in Reinforcement Learning}
The remarkable performance of LLMs in natural language processing tasks has inspired researchers to investigate their applicability in decision-making challenges~\cite{ahn2022can,huang2022language}. In~\cite{ahn2022can,huang2022language}, LLMs are applied for task decomposition and planning at a high level, while policies for lower-level execution are either learned or separately designed. Another line of research~\cite{li2022pre,lin2023text2motion,wang2023prompt} focuses on leveraging the representational and generalization capabilities of pre-trained LLMs. For instance,~\citet{li2022pre} fine-tune pre-trained LMs to generate policies for tasks where inputs can be transformed into word sequences, emphasizing the importance of input sequence structure;~\citet{lin2023text2motion} employ a geometric feasibility planner to guide the LM in producing both mid- and low-level plans based on language instructions, while~\citet{tang2023saytap} design prompts for encoding instructions into LMs.~\cite{reid2022can,shi2023unleashing} investigate the use of wikitext to enhance traditional offline RL tasks.

\section{Conclusion}
In this work, we introduced MDT4Rec, an offline reinforcement learning-based recommender system designed to address two key challenges: learning from sub-optimal user feedback and representing complex user-item interactions. By shifting trajectory stitching to the action inference stage, MDT4Rec allows for adaptive history truncation, ensuring that negative past experiences do not negatively impact decision-making. Additionally, by leveraging pre-trained large language models for initialization, replacing linear embedding layers with MLPs, and employing LoRA for efficient fine-tuning, MDT4Rec effectively enhances representation learning while maintaining computational efficiency. Extensive evaluations on public datasets and online simulations demonstrate that MDT4Rec consistently outperforms existing offline RL-based recommendation models. 

However, in its current form, MDT4Rec retains the language capabilities of LLMs but does not fully exploit them for recommendation tasks, partly due to the limited availability of relevant textual or user descriptions across datasets. A promising future direction is to explore how the language reasoning abilities of LLMs can be better leveraged to enhance offline RLRS performance. We believe that this study lays the groundwork for future research on effectively integrating LLMs within RLRS.
 
\section*{GenAI Usage Disclosure}
In accordance with the ACM authorship policy, no disclosure of generative AI usage is required for this work.
\bibliographystyle{ACM-Reference-Format}
\balance
\bibliography{sample}

\end{document}